\documentclass[sigconf]{aamas} 

\usepackage{balance} 

\usepackage{fancyhdr}
\usepackage{multienum}
\usepackage{epsfig}
\usepackage[tight]{subfigure}
\usepackage{wrapfig}
\usepackage{comment}
\usepackage{amsmath,amssymb,amsthm}
\usepackage{graphicx}
\usepackage{framed}
\usepackage{mathtools}
\usepackage{color}
\usepackage{mathrsfs} 
\usepackage{mathtools}
\usepackage[thinc]{esdiff}
\usepackage{natbib}
\setcitestyle{authoryear}
\usepackage{mathtools}
\usepackage{algorithm}
\usepackage{hyperref}
\newcommand{\Comments}{0}
\definecolor{gray}{gray}{0.5}
\definecolor{lightred}{rgb}{1,0.6,0.6}
\definecolor{darkgreen}{rgb}{0,0.5,0}
\definecolor{myorange}{rgb}{0.8,0.7,0.5}
\definecolor{darkblue}{rgb}{0.0,0.0,0.5}
\hypersetup{colorlinks=true,breaklinks=true,bookmarks=true,urlcolor=darkblue,citecolor=darkblue,linkcolor=darkblue,bookmarksopen=false,draft=false}

\usepackage[noend]{algpseudocode}

\newcommand{\R}{\mathbb{R}}

\renewcommand{\vec}[1]{{\mathbf{#1}}}
\renewcommand{\r}{\vec{r}}

\newcommand{\q}{\vec{q}}
\newcommand{\reals}{\mathbb{R}}

\newif\ifshowcomments
\showcommentstrue
\definecolor{darkgreen}{rgb}{0,0.75,0}
\newcommand{\docomment}[3]{\ifnum\Comments=1 \textcolor{#1}{[ #2 : #3 ]} \fi}
\newcommand{\maneesha}[1]{\docomment{blue}{Maneesha}{#1}}

\newcommand{\squishlist}{
   \begin{list}{$\bullet$}
    { \setlength{\itemsep}{0pt}      \setlength{\parsep}{3pt}
      \setlength{\topsep}{0pt}       \setlength{\partopsep}{0pt}
      \setlength{\leftmargin}{2.5em} \setlength{\labelwidth}{1em}
      \setlength{\labelsep}{0.5em} } }

\newcommand{\squishend}{
    \end{list}}

\setcopyright{ifaamas}
\acmConference[AAMAS '26]{Proc.\@ of the 25th International Conference
on Autonomous Agents and Multiagent Systems (AAMAS 2026)}{May 25 -- 29, 2026}
{Paphos, Cyprus}{C.~Amato, L.~Dennis, V.~Mascardi, J.~Thangarajah (eds.)}
\copyrightyear{2026}
\acmYear{2026}
\acmDOI{}
\acmPrice{}
\acmISBN{}

\acmSubmissionID{1220}

\author{Maneesha Papireddygari}
\affiliation{
  \institution{Boston College}
  \city{Boston}
  \country{U.S.A.}}
\email{papiredd@bc.edu}

\author{Xintong Wang}
\affiliation{
  \institution{Rutgers University}
  \city{New Brunswick}
  \country{U.S.A.}}
\email{xintong.wang@rutgers.edu}

\author{Bo Waggoner}
\affiliation{
  \institution{University of Colorado, Boulder}
  \city{Boulder}
  \country{U.S.A.}}
\email{bwag@colorado.edu}

\author{David M.\ Pennock}
\affiliation{
  \institution{Rutgers University}
  \city{New Brunswick}
  \country{U.S.A.}}
\email{dpennock@dimacs.rutgers.edu}

\begin{abstract}
Automated Market Makers (AMMs) are used to provide liquidity for combinatorial prediction markets that would otherwise be too thinly traded. 
They offer both buy and sell prices for any of the doubly exponential many possible securities that the market can offer. 
The problem of setting those prices is known to be \#P-hard for the original and most well-known AMM, the logarithmic market scoring rule (LMSR) market maker \cite{chen2008complexity}.
We focus on another natural AMM, the Constant Log Utility Market Maker (CLUM). Unlike LMSR, whose worst-case loss bound grows with the number of outcomes, CLUM has constant worst-case loss, allowing the market to add outcomes on the fly and even operate over countably infinite many outcomes, among other features. Simpler versions of CLUM underpin several Decentralized Finance (DeFi) mechanisms including the Uniswap protocol that handles billions of dollars of cryptocurrency trades daily.
We first establish the computational complexity of the problem: we prove that pricing securities is \#P-hard for CLUM, via a reduction from the model counting 2-SAT problem.
In order to make CLUM more practically viable, we propose an approximation algorithm for pricing securities that works with high probability. 
This algorithm assumes access to an oracle capable of determining the maximum shares purchased of any one outcome and the total number of outcomes that has that maximum amount purchased.
We then show that this oracle can be implemented in polynomial time when restricted to interval securities, which are used in designing financial options.
\end{abstract}


\newcommand{\BibTeX}{\rm B\kern-.05em{\sc i\kern-.025em b}\kern-.08em\TeX}

\begin{document}
\title{Efficiency of Constant Log Utility Market Makers}
%

\pagestyle{fancy}
\fancyhead{}


\maketitle 
\section{Introduction}


An \emph{Automated Market Maker} (AMM) is an algorithmic trading mechanism that serves as an alternative to the traditional continuous double auction mechanism ubiquitous in finance. AMMs intermediate every order between buyers and sellers, always offering some prices for traders to buy or sell and thus providing liquidity when order books would otherwise be thin or empty. As an extreme example, a single buyer cannot transact in an auction without a seller; on the other hand, every buyer of any security at any time will be offered some price by an AMM.
First introduced by \citet{hanson2003combinatorial}, AMMs have been extensively studied in theoretical research \citep{hanson2007logarithmic, chen2007utility,abernethy2013efficient}.
AMMs have been shown to be powerful aggregation mechanisms of crowd-sourced private information from heterogeneous agents \cite{ostrovsky2012information}.
In practice, AMMs power prediction markets like Manifold Markets and decentralized exchanges (DEXs) for cryptocurrencies including Uniswap, Balancer, and Curve. Uniswap alone has handled over \$3 trillion in total trade at a pace of over \$1 billion per day.\footnote{According to DeFiLlama ( \url{https://defillama.com/dexs/uniswap} ) and personal communication with a Uniswap Foundation team member.}

AMMs are especially useful in combinatorial settings where the number of distinct outcomes is exponential and the number of tradable assets (subsets of outcomes) is doubly exponential.
In this enormous sea of possibilities, a buyer can easily request a security that no seller has yet specifically offered to sell.
However, an AMM would immediately quote some price for the buyer. 
Of course, storing the full list of prices is intractable. 
The market must instead compute prices on demand and do so in a reasonable amount of time.
For example, in a combinatorial prediction market for the US Presidential election where one of two political parties can win in each of fifty states, there are $2^{50}$ possible outcomes and $2^{2^{50}}$ tradable securities (a security is a subset of outcomes, like the subset of electoral maps where ``The Democrats win in exactly two states that begin with the letter M'').
A second example, and one that we will return to later in the paper, is \emph{interval betting}, where traders predict which interval a one-dimensional statistic will fall into at some future date (similar to how financial options traders predict the value of a stock price relative to a strike price at an expiration date).
The number of outcomes (e.g., the number of possible future values of the stock price) is technically infinite; however, in practice, the outcome space will always be discretized. 
Still, the number of possible outcomes can be very high, and the market would like to run all operations (price quotes and trades) in time and space that is logarithmic in the number of outcomes. Equivalently, we can think of the problem size as the number of bits required to encode the statistic, and the number of outcomes as exponential in this problem size. In this interpretation, the market wants all operations to run in time polynomial in the problem size.

Straightforward AMM pricing algorithms run in time and space polynomial in the number of outcomes. In combinatorial settings, the number of outcomes is exponential, making such algorithms intractable.
One practically crucial question is whether prices can be computed efficiently: that is, logarithmic in the number of outcomes or polynomial time overall.
\citet{chen2008complexity} give mostly negative answers.
They prove that computing prices in a logarithmic market scoring rule (LMSR) AMM is \#P-hard in Boolean combinatorial settings (i.e., the US Presidential election example) and in permutation combinatorial settings (i.e., where outcomes are the possible permutations of a set, as in a horse race), even when the types of securities allowed to trade are severely limited.
They also propose an approximation algorithm for subset betting in the permutation setting. Their algorithm requires a number of steps inversely proportional to the approximation error. The resulting market maker has a loss bound that is linear (as opposed to logarithmic) in the number of outcomes and that is unbounded unless the approximation error is zero. 

\textit{Constant Logarithmic Utility Market Maker (CLUM): } \citet{chen2007utility} define an AMM that operates by always maintaining a constant expected utility with respect to its prior for some risk-averse utility function: in the case of CLUM, the logarithmic utility function. They show how to convert among constant-utility market makers, scoring-rule-based market makers, and cost-function-based market makers. All of these algorithms bound the worst-case loss of the market maker.

\citet{hanson2003combinatorial} designed LMSR, the first AMM, and proved that it has unique desirable properties including respecting certain statistical independencies (e.g., a wager on event $A$ conditional on event $B$ does not affect the price of $B$).
\citet{chen2007utility} showed that LMSR corresponds to a constant negative exponential utility market maker. 
They also introduce the constant logarithmic utility market maker (CLUM). One feature of CLUM is that it has constant worst-case loss---independent of the number of outcomes---whereas the loss bound of LMSR grows with the number of outcomes.
As a result, CLUM can be used to operate markets that add outcomes dynamically (for example, an options market that allows traders to create custom strike prices) and even handle countably infinite outcome spaces (for example, a market to predict the future top trending term on X which may be a newly invented term).

Recently, \citet{frongillo2024axiomatic} proved a new and surprising result: constant utility market makers are equivalent to \emph{constant function market makers} (CFMMs), a class of AMMs that were developed independently and are extensively used to create decentralized exchanges (DEX) on blockchains, handling trillions of dollars in trade.
CFMMs operate by declaring a potential function $\varphi$ and specifying the initial reserves held by the market.
A trade is accepted only if the value of this potential function remains constant before and after the trade.
CLUMs are equivalent to \emph{constant product market makers} (CPMMs), a specific class of CFMMs that are widely used in DEX due to their desirable theoretical properties \citep{schlegel2022axioms}.
Some prominent examples include Balancer and Uniswap.
Introduced by \citet{adams2020whitepaperv2}, Uniswap is the predominant decentralized exchange in the Ethereum ecosystem, with a potential function of $\varphi(q_1,q_2) = q_1\cdot q_2$.
The generalization of Uniswap to $N$ assets, implemented on blockchain as Balancer, is given by the potential function 
\begin{align*}
    \varphi(q_1,\cdots,q_N) = (q_1\cdots q_N)^\frac{1}{N}
\end{align*}
where $q_i$ is the quantity of asset $i$ held in reserves by the market maker.


\textit{Our contributions:} Given the growing significance of these mechanisms and more recently so in DEX, it is imperative to study the computational complexity of combinatorial CLUMs, as it pertains to their practicality in supporting derivative markets.
In this paper, we prove that pricing a security according to CLUM is \#P-hard when the outcome space is Boolean combinatorial like the election example.
Section \ref{sec:cpmm_complexity} argues that the problem remains hard even in a simplified scenario where only securities of the form ``disjunctions of two negative or positive events'' are traded. 
We prove this by showing a reduction from the model counting 2-SAT problem to this problem.
In Section \ref{sec:approx_C}, we then propose an efficient algorithm for pricing securities arbitrarily close to true prices with arbitrarily high probability.
This approximation algorithm assumes access to oracle knowledge of the maximum quantity sold of any outcome in the market and the number of such max-bought outcomes.
In Section \ref{sec:examples}, we show that in 
interval betting markets, 
the oracle knowledge needed in the approximation algorithm \ref{alg:cost_approx} is easily attainable.
We use an augmented self-balancing tree with lazy propagation to achieve this.
This algorithm is inspired by \citet{dudik2021logtime}, that leverages an augmented, self-balancing binary search tree to store and efficiently update key statistics of the market state. 
In Appendix \ref{sec:other_approx}, we give an alternative approximation algorithm when the ORACLE needed for Algorithm \ref{alg:cost_approx} takes unreasonable time.
Particularly when the securities traded do not fall into the aforementioned special cases.
The algorithm proposed here relies heavily on one proposed by \citet{ermon2013taming}, uses a Maximum a Posteriori (MAP) oracle, and gives a 64 factor constant approximation.
This guarantee is worse than our main algorithm, but may do well in practice.



\subsection{Background}


Let there be $n$ binary events and let $N=2^n$ be the total possible outcomes associated with these events.
In the presidential election example, the events are whether each state votes red or blue.
There are 50 events and $2^{50}$ outcomes/possibilities.
Let us first consider the unrestricted setting corresponding to trading $N$ Arrow-Debreu securities.
Each security $i\in \{1,\cdots, N\}$ pays \$1 when the outcome $i$ happens and \$0 otherwise.
In order for the market maker to be able to accept any trade, it needs to start out with some initial reserves.
Let this be $C_0$ cash or alternately with $C_0$ shares of each security.
This is also the worst-case loss of the market maker.

Let the state of the market be defined by total quantity of securities sold so far in the market $\q' = (q_1',\dots,q_N')$ and a pricing function or cost-function $C$~\citep{chen2007utility,abernethy2013efficient}.
If the trader requests a bundle of securities $\r = \q-\q' \in \reals^N$ from this market and makes the new state of the market to be $\q$, a cost-function market maker grants this bundle for a cost of \$ $C(\q)-C(\q')$.
The Constant Log-Utility Market Maker (CLUM) has a $C(\q)$ defined as the solution to the following equation:
\begin{equation}\label{eqn:invariant}
    \frac{1}{N} \cdot \Sigma_{j=1}^N \log(C(\q)-q_j) = \log(C_0)
\end{equation}
The instantaneous price of security linked to a single outcome $j$ is given by
\begin{align*}
    p_j &= -\diffp{C}{{q_j}}
        = \frac{\frac{1}{C(\q)-q_j}}{\sum_{i=1}^N \frac{1}{C(\q)-q_i}}
\end{align*}

Now, if we are in a more restricted market that only trades securities from the set $S$, price of a security $s\in S$ can be given as

\begin{align*}
    p_s &= \sum_{j\in s} p_j = \frac{\sum_{j\in s} \frac{1}{C(\q)-q_j}}{\sum_{i=1}^N \frac{1}{C(\q)-q_i}}\\
        &= \frac{\sum_{j\in s} \frac{1}{C(\q)-\sum_{j\in s',s'\in S}q_{s'}}}{\sum_{i=1}^N \frac{1}{C(\q)-\sum_{i\in s',s'\in S}q_{s'}}}\\
\end{align*}
The total amount paid out, $q_j$, when outcome j occurs is calculated by summing the purchased amounts across all securities $s'\in S$ that are evaluated to true when j is the realized outcome.

\section{Computational Complexity} \label{sec:cpmm_complexity}

Although calculations of instantaneous price have an exponential number of terms, it may or may not translate to computational complexity.
In this section, we will quantify ``how difficult is it to compute the price'' by characterizing it into a computational complexity class.
We first show that pricing securities in CLUM markets is \#P-hard and we characterize other constant utility market makers that also have this hardness result.
We prove this hardness result by considering a restricted market setting that we describe below.

\subsection{Setting}

Consider Boolean Betting markets, in which the agent can bet on Boolean formulas of any two events or their negations.
Showing a hardness result for this sub-case will be enough to say that pricing securities in the general case is at least as hard.

\begin{itemize}
    \item $n$ possible binary events represented by $\{\mathcal{A}_1,\cdots, \mathcal{A}_n\}$.
    \item $N=2^n$ possible outcomes represented by the set $\Omega$.
    \item Each security $\mathcal{F}_i$ is of the form $\mathcal{A}_{i_1}\lor\mathcal{A}_{i_2}$ or $\bar{\mathcal{A}}_{i_1}\lor\mathcal{A}_{i_2}$ or $\mathcal{A}_{i_1}\lor\bar{\mathcal{A}}_{i_2}$ or $\bar{\mathcal{A}}_{i_1}\lor\bar{\mathcal{A}}_{i_2}$, where $i_1,i_2\in \{1,\cdots,n\}, i_1\neq i_2$.
    \item An outcome $\omega\in\Omega$ satisfies formula $\mathcal{F}_i$ if it makes the formula true. For shorthand notation, let this be denoted as $\omega\in \mathcal{F}_i$.
    \item Let $k<N$ securities $\mathcal{F}_i$, for $i\in\{1,\cdots,k\}$ be purchased in quantity $q$ each\footnote{Note that this is a restricted setting of the general case where any amounts of securities can be purchased. Any complexity result shown in this special case extends to the general case.}. 
\end{itemize}
\subsection{Reduction}

Let us assume that we have access to a subroutine that computes the price of an indicator security i.e. a security that pays \$1 for a single outcome and \$0 for all other outcomes.
Consider the following algorithm that reduces the problem of 2-SAT model counting problem to pricing combinatorial CLUM using the aforementioned subroutine - \\
\qquad
\\
\textbf{Input:} 2-SAT solution finder subroutine and subroutine that calculates prices of indicator securities.
\\
\textbf{Output:} Number of solutions of 2-SAT formula $\mathcal{F}$.
\begin{algorithm}[H]
    \caption{Reduce \#2-SAT $\to$ CLUM}
    \label{alg:reduction}
    \begin{algorithmic}[1]
    \State Find assignment $\omega$ that satisfies the 2-SAT formula $\mathcal{F} = \mathcal{F}_1 \land \mathcal{F}_2 \land \cdots \land \mathcal{F}_k$. 
    \State If no solution,
    \State \quad Return 0;
    \State If there is a solution,
    \State \quad Call the subroutine for pricing indicator security $\omega$ and store output as $p_{\omega}$.   
    \State \quad Return $\lfloor\frac{1}{p_{\omega}}\rfloor$.
    \end{algorithmic}
\end{algorithm}

\textit{Explanation:} Finding one satisfying assignment of a 2-SAT formula is efficiently solvable \citep{krom1967decision}.
If the price of an indicator security was computed in polynomial time using the input subroutine, we can output the total number of 2-SAT solutions of $\mathcal{F}$ using the above reduction.
But since the 2-SAT model counting problem is \#P-hard \cite{valiant1979complexity} and the reduction itself runs in polynomial time, the subroutine to price securities is at least \#P-hard.

The final crucial step is to demonstrate that the algorithm's output, $\lfloor\frac{1}{p_{\omega}}\rfloor$, equals the number of solutions to $\mathcal{F}$.

\begin{proposition} \label{prop:solutions}
    The number of solutions of $\mathcal{F}$ is $\lfloor\frac{1}{p_{\omega}}\rfloor$.
\end{proposition}
\begin{proof}
    Let $C'(\q)$ be the value of cost function after all the $k$ trades.
    Here $\q=(q_1,\cdots, q_N)$ where $q_i$ denotes the total amount of securities purchased that pay out when outcome $i$ happens.
    \begin{align*}
        \frac{1}{p_{\omega}} &= \frac{\sum_{i=1}^N \frac{1}{C'(\q)-\sum_{i\in s',s'\in S}q_{s'}}}{ \frac{1}{C'(\q)-\sum_{\omega\in s',s'\in S}q_{s'}}}\\
        &= \frac{\sum_{i=1}^N \frac{1}{C'(\q)-\sum_{i\in s',s'\in S}q_{s'}}}{\frac{1}{C'(\q)-k\cdot q}}\\
    \end{align*}
    Since $\omega\in\mathcal{F}$, we have $\omega\in\mathcal{F}_1,\cdots,\omega\in\mathcal{F}_k$ and $q$ shares each of security $\mathcal{F}_i$ were purchased.
    This makes the term $\sum_{\omega\in s',s'\in S}q_{s'} = k+\cdots+k+0+\cdots+0 = k\cdot q$.
    We can further split the above summation in the numerator into outcomes that satisfy the formula $\mathcal{F}$ and outcomes that don't.
    Let $\mathcal{M}$ be the set of outcomes that satisfy the formula $\mathcal{F}$.
    \begin{align*}
        \frac{1}{p_{\omega}} &= \frac{\sum_{i=1}^N \frac{1}{C'(\q)-\sum_{i\in s',s'\in S}q_{s'}}}{\frac{1}{C(\q)'-k\cdot q}}\\
        &= |\mathcal{M}| + \sum_{i=\mathcal{M}}^N \frac{C'(\q)-k\cdot q}{C'(\q)-\sum_{i\in s',s'\in S}q_{s'}}\\
    \end{align*}
\end{proof}

For the computational complexity result to transfer via Algorithm \ref{alg:reduction}, it is enough to show that any arbitrary pricing problem in CLUM has a reduction from the general model counting 2-SAT problem.
This flexibility allows for a strategic choice of the parameter $q$.
Specifically, setting $q = C_0(2^n-1)$ simplifies the subsequent proof by ensuring that the second term in the summation is strictly less than one.

\begin{align*}
    \sum_{i=\mathcal{M}}^N \frac{C'(\q)-k\cdot q}{C'(\q)-\sum_{i\in s',s'\in S}q_{s'}} & \leq 2^n\cdot\frac{C'(\q)-k\cdot q}{C'(\q)-(k-1)\cdot q} \\
    & \leq 2^n\cdot\frac{C'(\q)-k\cdot q}{C'(\q)-k\cdot q + C_0(2^n-1)} \\
    & \leq 2^n\cdot\frac{1}{1 + \frac{C_0(2^n-1)}{C'(\q)-k\cdot q}} \\
    & < 2^n\cdot\frac{1}{1 + \frac{C_0(2^n-1)}{C_0}} \text{\qquad  from Lemma 
    \ref{lemma:notallC}}\\
    & < 2^n\cdot\frac{1}{2^n}  < 1
\end{align*}

The first step above is the result of the fact that every outcome $i\notin \mathcal{M}$ is purchased at most $(k-1)q$ times and there are at most $N-|\mathcal{M}|\leq 2^n$ such terms.

\begin{lemma} \label{lemma:notallC}
    $C'(\q) - k\cdot q < C_0$.
\end{lemma}
\begin{proof}
    Assume the contrary that $C'(\q) - k\cdot q \geq C_0$.
    From the Constant Logarithmic Utility Market Maker invariance condition given in eq \ref{eqn:invariant} we have, 
    \begin{align*}
        \sum_{j=1}^N \log C_0 &= \sum_{j=1}^N \log (C'(\q)-q_j)\\
        & > \sum_{j=1}^N \log (C'(\q)-k\cdot q) \\
        C_0 &> C'(\q)-k\cdot q\\
    \end{align*}
    This is a contradiction to our assumption. The penultimate step follows from the fact that $q_j\leq k\cdot q , \forall j$ and $q_j < k\cdot q$ for at least one $j$ as we assume that $k < N$.
\end{proof}
\subsection{General complexity result}
\label{sec: general_complexity}
Here we aim to characterize other Constant Utility Market makers for which pricing a security is \#P-hard.

\begin{proposition}
    If a Constant Utility Market Maker characterized by utility function $f(\q)$ satisfies the following conditions, it is \#P-hard to price outcomes and securities in that combinatorial market.
    \begin{enumerate}
        \item $f(\q)$ is additively separable and symmetric w.r.t. outcomes\maneesha{right wording?} i.e. $f(\q)=\sum_{i=1}^N \varphi(q_i)$.
        \item $\varphi$ is strictly concave and increasing function. \maneesha{Common for CFMMs}
        \item $(\varphi')^{-1}(x*y)=(\varphi')^{-1}(x)*(\varphi')^{-1}(y)$ or alternatively , $(\varphi')^{-1}(x) = x^{\ln((\varphi')^{-1}(e))}$ where $\varphi' = \partial\varphi$
        \item $(\varphi')^{-1}(\frac{1}{2^n}) > 1$.
    \end{enumerate}
\end{proposition}
\begin{proof}
    Reduction algorithm is similar to the one proposed in algorithm \ref{alg:reduction}. However, we still have to prove that its output corresponds to number of solutions of $\mathcal{F}$. Let $q_i$ be the total quantity of shares of outcome $i$ purchased and $C$ be the cost after $k$ transactions.
    Price of an outcome is given by -
    \begin{align*}
        p_i &= \frac{\partial f}{\partial q_i} = \frac{\frac{\partial \varphi}{\partial q_i}}{\frac{\partial \varphi}{\partial C}}\\
        &= \frac{\varphi'(C-q_i)}{\sum_{\omega} \varphi'(C-q_{\omega})}\\
        \frac{1}{p_i} &= \sum_{\omega}\frac{ \varphi'(C-q_{\omega})}{\varphi'(C-q_i)}\\
    \end{align*}
    \maneesha{can use some eyes on signs above}
    
    If $i$ is a satisfying assignment for the 2-SAT formula and there are $|\mathcal{M}|$ solutions,
    \begin{align*}
        \frac{1}{p_i} &= |\mathcal{M}| + \sum_{\omega \notin \mathcal{M}}\frac{ \varphi'(C-q_{\omega})}{\varphi'(C-kq)}\\
    \end{align*}
    Maximum of the last term is 
    \begin{equation} \label{eq:1}
        \sum_{\omega \notin \mathcal{M}}\frac{ \varphi'(C-q_{\omega})}{\varphi'(C-kq)} < 2^n\frac{\varphi'(C-(k-1)q)}{\varphi'(C-kq)}\\
    \end{equation}
    We use assumption 2 in the proposition that $\varphi$ is strictly concave aka $\varphi'$ is strictly decreasing.
    Strict inequality as we do this only when there is at least one solution to 2-SAT.

    

    Now assume that $q=C_0((\varphi')^{-1}(1/2^n)-1)$, which is non-negative from assumption 4 of the proposition.

    \begin{align*}
        \frac{\varphi'(C-(k-1)q)}{\varphi'(C-kq)} &= \varphi'(\frac{C-(k-1)q}{C-kq})\\
        &\leq \varphi'( 1+ \frac{q}{C-kq})\\
        &< \varphi'(1+\frac{C_0((\varphi')^{-1}(1/2^n)-1)}{C_0})\\
        &< \varphi'((\varphi')^{-1}(1/2^n)) = \frac{1}{2^n}\\
    \end{align*}
    
    First step follows from assumption 3 of the proposition.
    And penultimate step is from the assumption $\varphi'$ is strictly decreasing and the fact that $C-kq < C_0$.
    Putting this back in Equation \ref{eq:1}, we have that 
    \begin{align*}
         \sum_{\omega \notin \mathcal{M}}\frac{ \varphi'(C-q_{\omega})}{\varphi'(C-kq)} < 1
    \end{align*}
    And hence we can write $|\mathcal{M}| = \lfloor\frac{1}{p_\omega}\rfloor$. Rest of the proof continues as in proposition \ref{prop:solutions}.
\end{proof}

\begin{section}{Approximating $C(\q)$}
\label{sec:approx_C}
Given that we proved pricing in this market is hard, the natural question to ask is if we can approximate the instantaneous price or alternately, value of the cost-function in a reasonable way.
We make a few additional assumptions that are crucial for the approximation algorithm but are also reasonable. 
\begin{itemize}
    \item Let $q_{max}$, the maximum quantity purchased of any outcome, and the number of outcomes that have $q_{max}$ shares purchased can be computed by querying an ORACLE.
    \item We assume that shares of securities can only be purchased in integral quantities.
\end{itemize}

Although it is difficult to prove that every security market will have a reasonable oracle, the problem of finding $q_{max}$ exhibits a much more structure than pricing any security.
Here we give a few such examples and we dedicate the next section to explore one such market in depth.
\begin{itemize}
    \item If the securities are of the form disjunctions of two events or their negations, the oracle we need would be a weighted MAX-SAT oracle. There are several weighted MAX-SAT oracles that work very well in practice \cite{argelich2008first}.
    \item It is also easily computable when for example the market is trading only securities of one event or a disjunction of two events.
    \item We will show in Section \ref{sec:examples} that a polynomial runtime oracle can be achieved for interval betting.
\end{itemize}

To recap, for any prespecified $\q$, we are trying to estimate $C(\q)$ that satisfies equation \ref{eqn:invariant} restated below
\begin{equation*}
    \frac{1}{N} \cdot \Sigma_{j=1}^N \log(C(\q)-q_j) = \log(C_0) .
\end{equation*}
First, it will be useful to define the following.
\begin{itemize}
    \item $s_{qmax} = |\{ j : q_j = q_{\max}\}|$, the number of outcomes for which a maximum number of securities are purchased.
    \item $U_1(c) = \frac{s_{qmax}}{N} \log(c - q_{max})$.
    \item $U_2(c) = \frac{1}{N} \sum_{j : q_j < q_{max}} \log(c - q_j)$
    \item $U(c) = U_1(c) + U_2(c)$. Note that the above equality can we rewritten as $U(C(\q))=\log C_0$.
\end{itemize}

\begin{proposition}\label{prop:Cbounds}
    The solution $C(\q)$ to Equation \ref{eqn:invariant} satisfies \\
    $\max\{C_0, q_{max}\} \leq C(\q) \leq q_{max} + C_0$.
\end{proposition}
\begin{proof}
    $U(c)$ is a monotone increasing function in $c$.
    As $c \to q_{\max}$ from above, $U(c) \to -\infty$, hence $C(\q)$, the solution of equation \ref{eqn:invariant}, needs to be greater than $q_{\max}$.
    Each $q_j \geq 0, \forall j$ and since $U(C(\q))$ is an average of terms of the form $\log(C(\q) - q_j)$, we get $\log C(\q)\leq U(C(\q)) = \log(C_0)$. 
    This implies $C(\q) \geq C_0$.
    On the other hand, assuming $C(\q) > q_{max} + C_0$ implies that $\log (C(\q)-q_j) > \log (C(\q)-q_{max}) > \log C_0$.
    This makes equation \ref{eqn:invariant} fail and hence renders our initial assumption of $C(\q)>q_{max}+C_0$ untrue.
\end{proof}

We are now ready to propose our approximation algorithm \ref{alg:cost_approx}.\\

\textit{Explanation:}
The core of the algorithm is running binary search to find $C(\q)$ that satisfies $U(C(\q))=\log C_0$.
The search starts with known bounds of $C(\q)$ given by proposition \ref{prop:Cbounds}.
For each potential candidate $\hat{C}^t$, we can precisely calculate $U_1(\hat{C}^t)$ using ORACLE.
But $U_2(\hat{C}^t)$ cannot be calculated due to it potentially having exponential number of terms.
Hence, it is estimated by $\hat{U}_2(\hat{C}^t)$ using a sampling Algorithm \ref{alg:U2_approx} with guarantees we will prove in lemma \ref{lemma:U2hat-accurate}.
Comparison of $U_1(\hat{C}^t)+\hat{U}_2(\hat{C}^t)$ with $\log C_0$ informs us of which direction to proceed in the next iteration of the binary search.
Note that since these terms are estimates, it is possible to miss the correct $C(\q)$ and hence more rigorous theoretical guarantees need to be proven.
Theorem \ref{thm:correctC} proves that we get a good multiplicative approximation of $C(\q)$ with high probability and before we prove it, we will first prove the lemmas \ref{lemma:Trounds},\ref{lemma:U2hat-accurate}. 

\begin{algorithm}[H]
    \caption{Approximately solve for $C(\q)$ in $U(C(\q)) = \log(C_0)$.}
    \label{alg:cost_approx}
    \begin{algorithmic}[1]
    \State Choose error parameter $\epsilon$ and a confidence level $\delta$.
    \State Compute $q_{max},s_{qmax}$ using ORACLE.
    \State Let $a^1=\max \{q_{max},C_0\},b^1=C_0+q_{max}$. 
    \State Let $t=1$.
    \State While $b^t/a^t > e^{\epsilon}$:
    \State \quad Let $\hat{C}^{t} = \frac{a^t + b^t}{2}$.
    \State \quad Let $U_1^t = \frac{1}{N}\cdot s_{qmax} \cdot \log(\hat{C}^t - q_{max})$.
    \State \quad Get $\hat{U}_2^t \leftarrow $ Algorithm \ref{alg:U2_approx}.
    \State \quad Let $\hat{U}^t = U_1^t + \hat{U_2^t}$.
    \State \quad If $\hat{U^t} > \log(C_0)+\epsilon$ then:
    \State \quad\quad Let $a^{t+1} = a^t$.
    \State \quad\quad Let $b^{t+1} = \hat{C}^t$.
    \State \quad Else If $\hat{U}^t\leq \log C_0 -\epsilon$ then:
    \State \quad\quad Let $a^{t+1} = \hat{C}^t$.
    \State \quad\quad Let $b^{t+1} = b^t$.
    \State \quad Else: \label{line:else}
    \State \quad\quad Break. \maneesha{other option is to do nothing}
    \State \quad Set $t = t + 1$.
    \State Return $\hat{C}^t$.
  \end{algorithmic}
\end{algorithm}
\begin{algorithm}[H]
    \caption{Approximating $U_2(\hat{C}^t)$}
    \label{alg:U2_approx}
    \begin{algorithmic}[1]
    \State Given error,confidence parameters $\epsilon,\delta$, $q_{max} $,$s_{qmax}$, $\hat{C}^t$.
    \State Let $X_i=\log (\hat{C}^t-q_i)$.
    \State Let $L = \max\{0, \log(C_0 + q_{max})\}$.
    \State Let $T = \lceil \log_2(1/\epsilon) \rceil$.
    \State Set $m=\frac{T L^2 \cdot\log(2/\delta)}{2\epsilon^2}$ 
    \State Sample $m$ $X_i$s i.i.d. as long as $X_i\neq \log(\hat{C}-q_{max})$.
    \State Compute and return $\hat{U_2^t} =\frac{N-s_{qmax}}{N}\cdot\frac{\sum_{i=1}^m \hat{X}_i}{m}$.
    \end{algorithmic}
\end{algorithm}

\begin{lemma} \label{lemma:Trounds}
    Algorithm \ref{alg:cost_approx} terminates after at most $T = \lceil \log_2(1/\epsilon) \rceil$ iterations of the while loop.
\end{lemma}
\begin{proof}
    For $t=1$, we have $\frac{b^t}{a^t} = \frac{C_0 + q_{max}}{\max\{C_0,q_{max}\}} \leq 2$.
    Define the interval length as $\ell^t = b^t - a^t$.
    Binary search ensures that interval size is halved every round, i.e. $\ell^{t+1} = \frac{1}{2}\ell^t$.
    After $T = \lceil \log_2(1/\epsilon) \rceil$ rounds, we have $\ell^T = b^T-a^T = \frac{1}{2^T} \ell^1 \leq \epsilon \cdot \ell^1 = \epsilon \cdot \min\{C_0,q_{max}\}$.
    Therefore,
    \begin{align*}
        \frac{b^T}{a^T} &\leq \frac{a^T + \min\{C_0,q_{max}\}\epsilon}{a^T}  \\
          &= 1 + \frac{\min\{C_0,q_{max}\}\epsilon}{a^T}   \\
          &\leq 1 + \frac{\min\{C_0,q_{max}\}\epsilon}{\max\{C_0,q_{max}\}}  \\
          &\leq 1 + \epsilon  \leq e^{\epsilon} .
    \end{align*}
\end{proof}

For the remainder of the analysis, for simplicity, we use the following notations.
\begin{itemize}
    \item $U_2^t = U_2(\hat{C}^t)$.
    \item $U^t = U_1^t + U_2^t$.
\end{itemize}

\begin{lemma} \label{lemma:U2hat-accurate}
    With probability at least $1-\delta$, we have for all $t \geq 1$, $|\hat{U}_2^t - U_2^t| \leq \epsilon$ and $C(\q)\in [a^t, b^t]$.
\end{lemma}
\begin{proof}
    Let us first prove that the boundaries of samples $X_j$ drawn in the Algorithm \ref{alg:U2_approx} is $ [0,L]$ for all $j$, where $L=\max\{0, \log(C_0 + q_{max})\}$.
    From Algorithm \ref{alg:cost_approx},
    we have $\hat{C}^t \geq a^t \geq a^1 \geq q_{max}$.
    For any $j$ such that the quantity of security $j$ purchased $q_j \neq q_{max}$, $X_j = \log(\hat{C}^t - q_j) \geq \log(q_{max} - q_j) \geq \log(1) = 0$.
    This uses our assumption that shares are purchased in integer quantities.
    To prove the upper bound of $X_j$, consider the upper bound $\hat{C}^t \leq b^1 \leq C_0 + q_{max}$.
    Using this in the definition of $X_j$,
    $X_j = \log(\hat{C}^t - q_j) \leq \log(\hat{C}^t) \leq \log(C_0 + q_{max}) = L$.
    This established that $X_j \in [0,L]$.

    This boundedness property of $X_j$s enables us to apply Hoeffding's inequality,
    \begin{align*}
        \Pr[|U_2^t - \hat{U}_2^t| > \epsilon]
        &\leq 2e^\frac{-2m\epsilon^2}{L^2}  \leq \frac{\delta}{T} .
    \end{align*}
    By union bound over the at most $T$ rounds of the binary search, the probability of any round having greater than $\epsilon$ error is at most $\delta$.

    We now prove the claim that the true cost $C(\q)$ remains contained within the current interval, $C(\q)\in [a^t, b^t]$, using induction. 
    The base case, $t=1$, is true by the proposition \ref{prop:Cbounds}.
    Fix a time step $t$ and assume that it is true for $t-1$.
    Given that the above proven result holds for all time steps, $U_2^{t-1} - \hat{U}_2^{t-1}$ is bounded with probability $1-\delta$. 
    If the condition $ \hat{U}_2^{t-1} > \log C_0 +\epsilon$ is met in the algorithm, then $U^{t-1}>\log C_0$.
    This means that $\hat{C}^t = \frac{a^{t-1}+b^{t-1}}{2}>C(\q)$ and $C(\q) \in [a^{t-1},\frac{a^{t-1}+b^{t-1}}{2}] = [a^t,b^t]$.
    A similar argument can be made for when $ \hat{U}_2^{t-1} \leq \log C_0 -\epsilon$.
    The recursive step holds and hence concludes the proof.
\end{proof}
\begin{theorem}\label{thm:correctC}
    For a given $C_0,\q,\epsilon$, with probability of at least $1-\delta$, the outcome of Algorithm \ref{alg:cost_approx}, i.e. estimate $\hat{C}^T$ of $C(\q)$, satisfies 
    \begin{align*}
        \frac{C(\q)}{1+2\epsilon}&\leq \hat{C}^T\leq C(\q)(1+2\epsilon)
    \end{align*}
\end{theorem}

\begin{proof}

    When the algorithm terminates after $T$ iterations (as proven in lemma \ref{lemma:Trounds}) of the while loop, $\frac{b^T}{a^T}\leq e^\epsilon$.
    Combining this with Lemma \ref{lemma:U2hat-accurate} and the fact that $\hat{C}^T=\frac{a^T+b^T}{2}$, we can bound $\frac{\hat{C}^T}{C(\q)}$.
    \begin{align*}
        \frac{a^T+b^T}{2b^T}&\leq \frac{\hat{C}^T}{C(\q)} \leq \frac{a^T+b^T}{2a^T}\\
        \frac{e^{-\epsilon}\cdot b^T+b^T}{2b^T}&\leq \frac{\hat{C}^T}{C(\q)} \leq \frac{a^T+a^T\cdot e^\epsilon}{2a^T}\\
        \frac{e^{-\epsilon}+1}{2}&\leq \frac{\hat{C}^T}{C(\q)} \leq \frac{1+e^\epsilon}{2}\\
        \frac{1+\epsilon/2}{1+\epsilon}&\leq \frac{\hat{C}^T}{C(\q)} \leq 1+\frac{\epsilon}{2} \text{ \qquad  from Taylor expansion}\\
        \frac{1}{1+2\epsilon}&\leq \frac{\hat{C}^T}{C(\q)} \leq 1+2\epsilon
    \end{align*}
    
   If the algorithm terminates early by breaking the while loop at Line \ref{line:else}, 

    \begin{align*}
        |\hat{U}(\hat{C}^T)-\log C_0|\leq \epsilon \implies |U(\hat{C}^T)-\log C_0|\leq 2\epsilon 
    \end{align*}
    Expanding those terms, we get

\begin{align*}
    |\sum_i \left(\frac{1}{N}\log (C(\q)-q_i)-\frac{1}{N}\log (\hat{C}^T-q_i)\right)| &\leq 2\epsilon\\
    \frac{1}{N}|\sum_i \log \left(1+\frac{\hat{C}^T-C(\q)}{C(\q)^T-q_i}\right)| &\leq 2\epsilon \\
    \frac{1}{N}|\sum_i \log \left(1+\frac{\hat{C}^T-C(\q)}{C(\q)-0}\right)| &\leq 2\epsilon\text{   as $q_i\geq 0$.}\\
    \left|\log \left(\frac{\hat{C}^T}{C(\q)}\right)\right| &\leq 2\epsilon\\
    \left|\frac{\hat{C}^T-C(\q)}{C(\q)}\right| &\leq e^{2\epsilon} -1 \approx 2\epsilon \\
\end{align*}

\end{proof}

\textit{Addressing preserving the ``constant'' utility property}. The reader or practitioner could wonder if the approximation away from the true cost function leads to gradual drift or systemic error in the cost function value over successive trades. 
The algorithm design ensures that this does not occur. 
The error in the estimate of $C(\q)$ at, say time $t=10$, does not influence the calculation of $C(\q)$ at time $t=11$ as the algorithm recalculates the cost function from the invariance condition, equation \ref{eqn:invariant}, using the new $\q$. 
Practitioners implementing this algorithm will have reliable estimates of the instantaneous price, with the cost function estimate remaining in the $1+2\epsilon$ approximation of the true $C(\q)$ for \textbf{every} trade.

\end{section}

\section{Pricing Interval Securities with CLUM}\label{sec:examples}
In this section, we showcase the use of the proposed approximation algorithm to price interval securities in a prediction market.
We first formalize the problem of pricing a single interval security in a prediction market governed by a constant log utility market maker, and analyze its computational complexity.
We then demonstrate that for interval securities we have a poly-time oracle to query $q_{max}$, the maximum shares purchased for any outcome and $s_{qmax}$, the number of outcomes that has $q_{max}$.

\subsection{Setting}
We consider $N$ total possible outcomes in the world that are arranged in some order, and an interval security $I_J$ corresponds to a non-empty set of consecutive outcome indices $J \subseteq \{0, 1, ..., N-1\}$.
The purchase of one share of $I_J$ simply transforms the share vector $\q$ into $\q'$, where
\begin{equation*}
    q'_i = 
    \begin{cases}
    q_i + 1 &\text{if $i \in J$,}
    \\
    q_i &\text{otherwise.}
    \end{cases}
\end{equation*}
The market adopts a constant log utility market maker, of which the cost function is implicitly defined as the following
\[\frac{1}{N} \sum_{j=1}^N \log(C(\q)-q_j) = \log(C_0)\]
where $C_0$ is some constant initial cost, and the price of one share of $I_J$ is calculated as $C(\q')-C(\q)$.
The problem lies in computing $C(\q')$, i.e. solving equation \ref{eqn:invariant}, which we recall requires solving for $C$ in
\[\sum_{j=1}^N \log(C-q'_j) = K, \text{      where  } K = N \cdot \log(C_0).\]

\subsection{Hardness of Pricing Intervals Using CLUM}
Here we show that, pricing interval securities in CLUM market is hard, i.e. exponential in number of event $n$ or polynomial in number of outcomes $N$. 
The intuition lies in the fact that the cost function of the constant log utility market maker is globally dependent. 

\begin{proposition}
    Given only query access to entries of $\q$, pricing a single interval security using a Constant Log Utility Market Maker has a computational complexity of $\Omega(N)$, where $N$ is the total number of atomic outcomes.
\end{proposition}
\begin{proof}
    
We prove by contradiction, assuming that there exists a deterministic algorithm $\mathcal{A}$ that correctly computes the price of an interval security with $o(N)$ queries to the entries of $\q$.
That is, let $T(\mathcal{A}, \q, J)$ be the number of operations performed by $\mathcal{A}$ on a given input, then $\max_{\q, J}T(\mathcal{A}, \q, J) < N$.
Therefore, there is at least an index $j$ such that $q_j$ is not read by $\mathcal{A}$.

We construct two input instances $\q^A$ and $\q^B$ that are identical to each other except at index $j$, i.e., 
\begin{equation*}
    q^B_i = 
    \begin{cases}
    q^A_i &\text{if $i \neq j$,}
    \\
    q^A_i + \delta &\text{if $i = j$, for some $\delta > 0$.}
    \end{cases}
\end{equation*}
For simplicity, we let the purchased interval be $J=\{1,...,N\}$, representing the purchase of a security that pays out on any outcome.
Thus, we have $\q'^A = \q^A+\mathbf{1}$ and $\q'^B = \q^B+\mathbf{1}$
Let $P^A$ and $P^B$ be the correct prices for inputs $(\q^A, J)$ and $(\q^B, J)$.

Given our assumption of sublinear time, when algorithm $\mathcal{A}$ runs on input $(\q^B, J)$, it follows the same executions as for input $(\q^A, J)$, as all the input values it reads are identical, i.e., $q^A_i = q^B_i$ for $i \neq j$.
Thus, $\mathcal{A}$ must compute the same price, $P_\mathcal{A}^A = P_\mathcal{A}^B$, for both input instances.

Now we show that the true prices are different.
Let $C'^A$ be the root of the equation:
\[\sum_{i=1}^N \log(C-(q^A_i+1)) = K.\]
Let $C'^B$ be the root of the equation:
\[\log(C-(q^A_j+\delta+1))+\sum_{i\neq j} \log(C-(q^A_i+1)) = K.\]
For any valid $C > q_{\max}$, $\log(C-(q^A_j+\delta+1)) < \log(C-(q^A_j+1))$, thus $C'^B > C'^A$ in order to satisfy the equality to $K$.
The same thing holds for the initial costs, i.e., 
$C^B > C^A$.
Since the cost function $C$ is non-linear and strictly concave, the change in a single entry $q_j$ has a non-linear effect and the difference of the two costs before and after purchasing $J$ can not be the same, i.e., $P^A \neq P^B$.

This leads to the contradiction, and thus there exists \textit{no} such sublinear algorithm that can correctly compute the price of interval securities under CFMM.
\end{proof}

\subsection{Computing $q_{\max}$ and Number of Outcomes with $q_{\max}$}

To achieve performance that is independent of the outcome universe size $N$, we use an augmented, self-balancing binary search tree (BST). 
This structure only creates nodes at the endpoints of purchase intervals, making it significantly more efficient in both time and space. 

\subsubsection{Node and Tree Structure.}
The data structure is a self-balancing BST (e.g., an AVL or Red-Black Tree) where each node represents a unique endpoint of a previously purchased interval.
Each node is annotated with the following basic information:
\begin{itemize}
    \item \texttt{key}: The integer value of the endpoint.
    \item \texttt{value}: The share count for the elementary interval that begins at this node's key.
    \item \texttt{left\_child} and \texttt{right\_child}: These are pointers to child nodes, initialized to none.
\end{itemize}
Each node is also annotated with the following augmented data:
\begin{itemize}
    \item \texttt{max\_val}: The maximum number of shares held for any elementary interval in the subtree.
    \item \texttt{max\_count}: The total number of \textit{atomic outcomes} that have the \texttt{max\_val} in the subtree.
    \item \texttt{lazy\_add}: A pending increment to be applied to all nodes in this subtree.
\end{itemize}

\subsubsection{Initialization.}
The tree starts empty. 
This represents a single, infinite elementary interval covering all possible outcomes. 
Then we will create the very first node, and the tree will grow dynamically from there as more unique interval endpoints are introduced through purchases. 

\subsubsection{Update after an interval purchase.}
The algorithm (Algorithm~\ref{alg:bst_complete_full}) operates using a ``split-update-merge'' strategy to efficiently apply changes when a specific interval $[l, r]$ is purchased. 
First, it ensures that the interval's precise boundaries exist within the tree by inserting $l$ and $r+1$ as nodes if they are not already present. 
This step carves out the exact elementary intervals that will be affected. 
Next, it uses two \textit{split} operations to break the tree into three parts: all intervals before $l$, all intervals within $[l, r]$, and all intervals after $r$. 
The \textit{update} itself can be performed efficiently by applying a single lazy value to the root of the middle tree that represents the target range. 
Finally, two \textit{merge} operations combine the three parts back into a single, balanced tree. 

We note that throughout this process, from splitting and merging to rebalancing rotations, the augmented data (\texttt{max\_val, max\_count}) in each node is continuously recalculated based on itself and its two children, ensuring the global maximum at the root remains correct at all times.

\subsubsection{Query.}
Query the maximum quantity purchased \texttt{root.max\_val} and the number of outcomes at this maximum \texttt{root.max\_count} are $O(1)$ operations.

Let $U$ be the number of distinct interval purchases made and $k$ be the number of unique endpoints $(k
\le 2U)$.
Overall, we will need $O(k)$ to store the intervals, and it takes $O(\log k)$ for each interval purchase.
Specifically, the logarithmic time complexity is a direct result of using a self-balancing BST. 
The height of such a tree is always guaranteed to be proportional to the logarithm of the number of nodes, $k$.

\texttt{EnsureEndpointExists} involves a search ($O(\log k)$), potentially an insertion ($O(\log k)$), and a rebalance ($O(\log k)$). The total time is therefore $O(\log k)$.
\texttt{Split} and \texttt{Merge} are standard operations on balanced trees that work by traversing a path from the root. Their complexity is determined by the height of the tree, making them $O(\log k)$ as well.
\texttt{PushDown} and \texttt{UpdateAugmentedData} operate on a single node and its immediate children, taking constant, $O(1)$, time per call. They are called during the logarithmic-time path of the other operations, so they do not increase the overall complexity.
Since the main \texttt{IntervalPurchase} procedure consists of a fixed number of these $O(\log k)$ operations, its total time complexity remains $O(\log k)$.

\begin{algorithm}[t]
\caption{Interval market: $q_{max}$ and frequency of $q_{max}$ supported by an augmented self-balancing BST}
\label{alg:bst_complete_full}
\begin{algorithmic}[1]

\Procedure{IntervalPurchase}{$tree, l, r, val$}
    \State \Call{EnsureEndpointExists}{$tree, l$}
    \State \Call{EnsureEndpointExists}{$tree, r + 1$}
    
    \State $(T_{left}, T_{mid\_and\_right}) \leftarrow \Call{Split}{tree, l}$
    \State $(T_{mid}, T_{right}) \leftarrow \Call{Split}{T_{mid\_and\_right}, r+1}$
    
    \If{$T_{mid}$ is not empty}
        \State $T_{mid}.root.lazy\_add \leftarrow T_{mid}.root.lazy\_add + val$
    \EndIf
    
    \State $tree \leftarrow \Call{Merge}{T_{left}, T_{mid}}$
    \State $tree \leftarrow \Call{Merge}{tree, T_{right}}$
\EndProcedure
\end{algorithmic}
\end{algorithm}

\begin{algorithm}[t]
\caption*{\textbf{Supporting functions} for augmented self-balancing BST}
\label{alg:support_bst_complete_full}
\begin{algorithmic}[1]
\small
\Function{EnsureEndpointExists}{$tree, key$}
    \Comment{Insert key if not there.}
    \State $p \leftarrow \Call{FindPredecessorNode}{tree, key}$ 
    \If{$p$ is null \textbf{or} $p.key \neq key$}
        \State $inherited\_value \leftarrow (p \text{ is not null}) \text{ ? } p.value : 0$
        \State \Call{Insert}{$tree, key, inherited\_value$} \Comment{BST insertion}
        \State \Call{Rebalance}{tree}
    \EndIf
\EndFunction
\vspace{1em}

\Function{PushDown}{$node$}
    \Comment{Push pending lazy value down}
    \If{$node$ is null \textbf{or} $node.lazy\_add = 0$} \State \Return \EndIf
    
    \State $node.value \leftarrow node.value + node.lazy\_add$
    \State $node.max\_val \leftarrow node.max\_val + node.lazy\_add$
    
    \If{$node$ is not a leaf}
        \If{$node.left\_child$ is not null}
            \State $node.left\_child.lazy\_add \leftarrow$ \State $node.left\_child.lazy\_add + node.lazy\_add$
        \EndIf
        \If{$node.right\_child$ is not null}
            \State $node.right\_child.lazy\_add \leftarrow$ \State $node.right\_child.lazy\_add + node.lazy\_add$
        \EndIf
    \EndIf
    
    \State $node.lazy\_add \leftarrow 0$
\EndFunction
\vspace{1em}

\Function{UpdateAugmentedData}{node}
    \Comment{Recalculates a node's augmented data from its children.}
    \State \Call{PushDown}{$node.left\_child$} 
    \State \Call{PushDown}{$node.right\_child$}
    \State $mval \leftarrow node.value$
    \If{$node.left\_child$} 
    \State $mval \leftarrow \max(mval, node.left\_child.max\_val)$ \EndIf
    \If{$node.right\_child$} 
    \State $mval \leftarrow \max(mval, node.right\_child.max\_val)$ \EndIf
    \State $node.max\_val \leftarrow mval$, 
    $mcount \leftarrow 0$
    \If{$node.value=mval$}
        \State $next\_key \leftarrow \Call{SuccessorKey}{node}$
        \State $mcount \leftarrow mcount + (next\_key - node.key)$
    \EndIf
    \If{$node.left\_child$ \textbf{\&} $node.left\_child.max\_val=mval$}
        \State $mcount \leftarrow mcount + node.left\_child.max\_count$
    \EndIf
    \If{$node.right\_child$ \textbf{\&} $node.right\_child.max\_val=mval$}
        \State $mcount \leftarrow mcount + node.right\_child.max\_count$
    \EndIf
    \State $node.max\_count \leftarrow mcount$
\EndFunction

\vspace{1em}
\Function{Merge}{$T_{left}, T_{right}$} \Comment{Merge using a max-key pivot} \EndFunction
\Function{Split}{node, split\_key} \Comment{Standard recursive implementation} \EndFunction
\Function{Rebalance}{node} \Comment{AVL/Red-Black rotations} \EndFunction

\end{algorithmic}
\end{algorithm}

\subsubsection{Approximating $C(\q)$} Given the oracle provided above to obtain $q_{\max}$ and $s_{qmax}$, we can apply Algorithm~\ref{alg:cost_approx} to approximate the cost $C(\q)$. 
Finding the exact root would otherwise take at least $\Omega(N)$, where $N$ can be huge representing the total number of atomic outcomes.
For example, if we use Newton's method to iteratively find the root and let $\epsilon$ denote the desired precision, it requires $O(N \log(1/\epsilon))$, where the number of iterations to converge can depend logarithmically on $1/\epsilon$.

\section{Conclusion and Future directions}
In this paper, we address some of the fundamental questions regarding the efficiency of constant log utility market makers.
We show that pricing of the securities in this market falls under the \#P-hard complexity class.
To the best of our knowledge, we give the first approximation algorithm for this combinatorial market with $\epsilon-\delta$ guarantees, conditional on having access to an oracle of computing $q_{max},s_{qmax}$.
We demonstrate a crucial application: market trading of interval securities.
We show that the aforementioned oracle can be implemented in poly-time using an augmented self-balancing binary search tree, effectively rendering the overall pricing problem tractable for this important market structure.

\paragraph{Future directions} 
A promising direction is to identify and characterize other structured security classes (beyond interval securities and ones discussed in Section \ref{sec:approx_C}) where an oracle can be implemented efficiently. 
Another open direction is to improve the 64-approximation achieved in Appendix \ref{sec:other_approx}.
Mispricing by a multiplicative factor of 64 can be catastrophic from the market maker's perspective.
Progress in this direction can be helpful when securities do not belong to a well defined class or when designing oracle is computationally hard.

\section{Acknowledgments} We would like to sincerely thank anonymous reviewers for their helpful comments.



\bibliographystyle{ACM-Reference-Format} 
\bibliography{modelcounting}
\pagebreak
\appendix
\section{Alternative approximation scheme} \label{sec:other_approx}
As we have seen that pricing in the market is hard, the natural question to ask is if we can approximate the solution in a reasonable way.
But as we see in Algorithm \ref{alg:reduction}, there is a reduction from  approximating \#SAT to an instance of this market i.e. the problem of approximation is not in FPTAS . 

Finding price of a security is quite close to the problem of finding weighted model counting in SAT problems.
While there is a lot of literature that caters to heuristics for solving model counting and weighted model counting, existing research doesn't have algorithms that provide theoretical guarantees.
One of the exceptions of this being the WISH algorithm proposed by \cite{ermon2013taming}.

Let the security that needs pricing be $S$. To achieve theoretical guarantees we need this security to be of the form of disjunction of two events. We modify the WISH algorithm slightly to be able to price such $S$. 

\textbf{Input:} A weight function $w:\Sigma\to\R^+$, $n=\log_2 \Omega$, $\delta, \epsilon, S$.
Here $\delta$ is the confidence level of the approximation of the answer.

\textbf{Output:} Estimate of price of security $S$.
\begin{algorithm}[H]
    \caption{WISH}
    \label{alg:wish}
    \begin{algorithmic}[1]
    \State T = $\lceil \frac{\ln(n/\delta)}{\alpha}\rceil$; $k$ 
    \State for $i = 0,1,\cdots, n$ do
    \State \quad for $t = 1,\cdots, T$ do
    \State \quad\quad Sample uniformly a hash function $h_{A,b}^i:\Omega\to \{0,1\}^n$.
    \State\quad\quad $w_i^t = \max \{kmax_\omega w(\omega)$ such that $(A\omega = b) \bigcap S$ \}.
    \State\quad\quad $w_i^{t'} = \max \{kmax_\omega w(\omega)$ such that $(A\omega = b) \bigcap (\neg S)$\}.
    \State\quad end for
    \State\quad $M_i= \textit{Median}(w_i^1,\cdots,w_i^T)$
    \State\quad $M_i'= \textit{Median}(w_i^{1'},\cdots,w_i^{T'})$
    \State end for
    \State Let $N = M_0+\Sigma_{i=0}^{n-1}M_{i+1}\cdot2^i$, $N'= M_0'+\Sigma_{i=0}^{n-1}M_{i+1}'\cdot2^i$.
    \State Return $\frac{N}{N+ N'}$
    \end{algorithmic}
\end{algorithm}
In the above algorithm, $kmax$ is a $k$-MAP oracle call.
We do this as appropriate choice of $k$ increases the probability of finding a non-trivial max element that belongs to $S$ and $\neg S$.

We prove upper and lower bounds for the price to lie in the range we predict using WISH algorithm.
We follow the proof by \cite{ermon2013taming} very closely except we try to tighten the lower bound of $N$ and tighten the upper bound of $N'$.
This way we can achieve a closer lower bound on price which is desirable to lower bound the loss the market maker will take on when pricing securities.

Fix an ordering of $\omega_i$ s such that $w(\omega_j)\geq w(\omega_{j+1})$ for all $1\leq j\leq 2^n$.
Define $b_i = w(2^i)$ and bin $B_i \sim \{\omega_{2^i+1},\cdots,\omega_{2^{i+1}}\}$.
Note that bin $B_i$ has $2^i$ outcomes.
\begin{lemma}\label{lemma:unequal-limits}
    For any $c \geq 2$ and $k=12$, there exists $\alpha^* >0 $ such that for $0<\alpha\leq\alpha^*$,
\begin{equation*}
    Pr[M_i'\in[b_{\min\{i+c,n\}},b_{\max\{i-1,0\}}]]\geq 1-\exp{(-T\alpha)}
\end{equation*}
\end{lemma}
\begin{proof}
    $w_{ij}$ are the k-max elements.
    From Lemma 1 of \cite{ermon2013taming}, we can say the following -
    \begin{align*}
        Pr[w_i' \geq b_{i+c}] &= 1- Pr[(w(\sigma_{ij}') \leq b_{i+c} \lor \sigma_{ij}' \notin \neg S) , \forall j \in [0,k]]\\
        &= 1- Pr[(w(\sigma_{ij}') \leq b_{i+c} \lor \sigma_{ij}' \notin \neg S) ]^k\\
        &= 1- \left(1-Pr[w(\sigma_{ij}') \geq b_{i+c}]\cdot Pr[\sigma_{ij}' \notin S ], \text{some } j\right)^k\\
        &\geq 1- \left(1-\frac{5}{9}\cdot (\frac{3}{4})\right)^k\\
        & \geq 0.52387 > 0.5
    \end{align*}
    for $c\geq 2$.
    And for the upper bound, lets define \\$S_j(h^i)=\sum_{\{\sigma\in\mathcal{X}_j\}} \mathbf{1}_{A\sigma=b (mod 2)}$ where $\mathcal{X}_j$ is set of $2^j$ heaviest outcomes.
    \begin{align*}
        Pr[w_i' \leq b_{i-1}] &= Pr[w_{ij} \leq w(\sigma_{2^{i-1}}) \forall j\in[0,k]] \\
        &+ Pr[w_{ij}' \in S , \forall j \in [0,k]]\\
        &= Pr[w_{i1} \leq w(\sigma_{2^{i-1}}) ] + Pr[w_{ij}' \in S , \forall j \in [0,k]]\\
        &= Pr[S_{i-1}(h^i)=0] + (\frac{3}{4})^k\\
        &= (1-\frac{1}{2^i})^{2^{i-1}}+ (\frac{3}{4})^k\\
        &\geq 0.5 + 0.03167 = 0.53167\\
    \end{align*}
The last inequality is because $(1-\frac{1}{2^i})^{2^{i-1}}$ is monotonically increasing in $i$ and is always greater than equal to 0.5.

From these, using Chernoff inequality and Hausdorff inequalities, we can say that 
\begin{align*}
    Pr[M_i'\leq b_{i-1}] &\geq 1-\exp^{-\alpha_1 T}\\
    Pr[M_i'\geq b_{i+c}] &\geq 1-\exp^{-\alpha_2 T}\\
\end{align*}
where $M_i'$ is the median of $\{w_i^{1'},\cdots,w_i^{T'}\}$, $\alpha_1=2(0.52387-0.5)^2 = 0.0011$ and $\alpha_2=2(0.53167-0.5)^2 = 0.002$.

Putting them together, 
\begin{align*}
    Pr[b_{i+c}\leq M_i' \leq b_{i-1}] &\geq 1-\exp^{-\alpha_1 T}-\exp^{-\alpha_2 T}\\
    &\geq 1-2\exp^{-\alpha_1 T}\\
    &\geq 1-\exp^{-\alpha^* T}
\end{align*}
where $\alpha^*=\ln2 \alpha_1 = 0.000762$
\end{proof}
Similarly we can hope to show that

\begin{lemma} \label{lemma:limits}
    Let $L'=b_0+\sum_{0}^{n-1} b_{\min\{i+c+1,n\}}2^i$ and $U'=b_0+\sum_{0}^{n-1} b_{\max\{i,0\}}2^i$. Then $U'\leq 2^{c+1}L'$
\end{lemma}
\begin{proof}
    \begin{align*}
        L'&= b_0 + \sum_{i=0}^{n-c-2} b_{i+c+1}2^i + \sum_{n-c-1}^{n-1} b_{n}2^i\\
        &= b_0 + \sum_{i=0}^{n-c-2} b_{i+c+1}2^i + \sum_{n-c-1}^{n-1} b_{n}2^i\\
    \end{align*}
    \begin{align*}
        U'&= b_0 + \sum_{i=0}^{0} b_{0}2^i + \sum_{i=1}^{n-1} b_{i}2^i\\
        &= 2\cdot b_0 + 2\cdot \sum_{i=1}^{n-1} b_{i}2^{i-1}\\
        &= 2\cdot b_0 + 2( \sum_{i=1}^{c} b_{i}2^{i-1}+\sum_{i=c+1}^{n-1} b_{i}2^{i-1})\\
        &\leq 2\cdot b_0 + 2( \sum_{i=1}^{c} b_{0}2^{i-1}+\sum_{i=c+1}^{n-1} b_{i}2^{i-1}) \\
        &\leq 2\cdot 2^c\cdot b_0 + 2 \cdot 2^c( \sum_{i=c+1}^{n-1} b_{i}2^{i-1-c})\\
        &\leq 2^{c+1} \left( b_0 +  \sum_{i=c+1}^{n-1} b_{i}2^{i-1-c}\right)\\
        &\leq 2^{c+1} \left( b_0 +  \sum_{i=c+1}^{n-1} b_{i}2^{i-1-c}+\sum_{i=n-c-1}^{n-1} b_{n}2^{i}\right)\\
        &\leq 2^{c+1} L'
    \end{align*}
\end{proof}

\begin{lemma}
    For any $c \geq 2$ and $k=12$, there exists $\alpha^* >0 $ such that for $0<\alpha\leq\alpha^*$,
\begin{equation*}
    Pr[M_i\in[b_{\min\{i+1,n\}},b_{\max\{i-c,0\}}]]\geq 1-\exp{(-T\alpha)}
\end{equation*}
\end{lemma}
And
\begin{lemma}
    Let $L=b_0+\sum_{0}^{n-1} b_{\min\{i+1+1,n\}}2^i$ and $U=b_0+\sum_{0}^{n-1} b_{\max\{i-c,0\}}2^i$. Then $U\leq 2^{c+1}L$
\end{lemma}

Proofs of these lemmas can be looked up from \cite{ermon2013taming}.
\begin{theorem}
    For any $\delta>0$ and $\alpha \leq 0.000762$, Algorithm \ref{alg:wish} makes $\Theta(n\ln n/\delta)$ k-MAP queries and, with probability atleast $\delta$, outputs a 64-approximation of price of security $S$.
\end{theorem}
\begin{proof}
    We can see from Lemma \ref{lemma:limits} and \ref{lemma:unequal-limits} that $N\in [L,U], N'\in [L',U']$, where $U'\leq 8 L', U\leq 8L$ with probability atleast $1-\delta$ when $T=\lceil\frac{\ln(n/\delta)}{\alpha}\rceil$.
    This implies that price of security $S$ is bounded below by $\frac{L}{L+8L'}$ and bounded above by $\frac{8L}{8l+L'}$ and hence it gives a 16-approximation.
\end{proof}
\end{document}